\documentclass[aps,prl,reprint]{revtex4-1}
\usepackage{blindtext}
\usepackage{graphicx}% Include figure files
\usepackage{dcolumn}% Align table columns on decimal point
\usepackage{bm}% bold 
\usepackage{amssymb}
\usepackage{color}
\usepackage[T1]{fontenc}% 
\usepackage{mathrsfs,amsfonts,dsfont}
\usepackage{nccmath}
\usepackage{amstext}
\usepackage{mathtools}
\usepackage{enumitem}
\graphicspath{{graphics/}}
\usepackage{rotating} % for diagram- descriptions
\usepackage[english]{babel}  

\usepackage{wrapfig}
\usepackage{lipsum} 

\usepackage{braket}
\usepackage{bbm} 
\usepackage[colorlinks=true,citecolor=blue,linkcolor=blue]
{hyperref}

\usepackage{amsthm}
\newtheorem{theorem}{Theorem}
\newtheorem{Lemma}{Lemma}
\newtheorem*{theorem2}{Theorem 2}
\newtheorem*{theorem3}{Theorem 3}
\newtheorem*{theorem1a}{Theorem 1}

\begin{document}
\title{
Quantifying necessary quantum resources for nonlocality}
\author{Lucas Tendick}
\email{Lucas.Tendick@hhu.de}
\author{Hermann Kampermann }
\author{Dagmar Bru\ss}
\affiliation{Institute for Theoretical Physics III, Heinrich-Heine-Universit\"at D\"usseldorf, D-40225 D\"usseldorf, Germany}

\begin{abstract}
Nonlocality is one of the most important resources for quantum information protocols. The observation of nonlocal correlations in a Bell experiment is the result of appropriately chosen measurements and quantum states. We quantify the minimal purity to achieve a certain Bell value for any Bell operator. Since purity is the most fundamental resource of a quantum state, this enables us also to quantify the necessary coherence, discord, and entanglement for a given violation of two-qubit correlation inequalities. Our results shine new light on the CHSH inequality by showing that for a fixed Bell violation an increase in the measurement resources does not always lead to a decrease of the minimal state resources.
\end{abstract}
\maketitle
It is arguably one of the most astonishing features of quantum theory that local measurements performed on certain quantum states can lead to the phenomenon of quantum nonlocality \cite{Nonlocality_review}. That is, the measurement statistics cannot be explained classically as they are not compatible with the principle of local realism. Mathematically this can be witnessed by the violation of a so-called Bell inequality \cite{Bell_seminal}. Even though nonlocality has been studied ever since the foundations of quantum theory \cite{EPR_paper}, it is not yet completely understood. \\ In this Letter, we address the following fundamental question: what are the required properties of a quantum state and its measurements in order to exhibit nonlocality? In other words, we quantify the interplay between the resource of quantum nonlocality \cite{Nonlocality_resource} and other resources like purity \cite{Purity_resource}, coherence \cite{Coherence_Aberg,RevModPhys_coherence}, discord \cite{Discord}, and most famously entanglement \cite{RevModPhys.81.865} on the state side and measurement incompatibility \cite{Incomop_resource} on the measurement side. The physical situation we are going to consider is illustrated in Fig. \ref{Introduction_figure}. \\ 
\indent Historically, the interplay between nonlocality and entanglement is the most studied one and therefore the best established one among these concepts. It is for instance well known that the existence of entanglement is necessary and for pure states even sufficient \cite{GISIN_pure} for quantum nonlocality while this does not hold true for mixed states \cite{PhysRevA.40.4277}. \\ More recently, the importance of the other mentioned quantum resources found growing attention  \cite{PhysRevA.77.042303,PhysRevLett.104.080501,Zhu2017,PhysRevLett.115.020403,PhysRevLett.114.210401,PhysRevA.102.012420,PhysRevA.97.012129,Bene2018}. Notably, for quantum states and measurements qualitative hierarchical relations between these resources have been established \cite{Purity_resource,PhysRevA.93.052112}. 
However, we are still in need of exact quantitative relations, for a complete view of the grand scheme. \\
Here, we are taking a step towards this aim by deriving from the spectrum of any Bell operator an analytical expression for the minimal purity of a quantum state that is needed to achieve some fixed amount of nonlocality in terms of a Bell inequality violation. This result is general, i.e. it holds for any dimension, any number of parties, measurement settings, and outcomes. 
\begin{figure}[h!]
\includegraphics[scale =0.45]{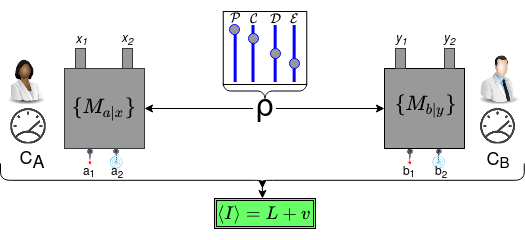}
 \caption{Illustration of a Bell experiment. A quantum state $\rho$ with adjustable resources purity $\mathcal{P}$, coherence $\mathcal{C}$, discord $\mathcal{D}$, and entanglement $\mathcal{E}$ is distributed to Alice and Bob who perform measurements $\lbrace M_{a \vert x} \rbrace$ and $\lbrace M_{b \vert y} \rbrace$ with also adjustable incompatibilities $C_A$ and $C_B$. The interplay between the state resources and the measurement resources results in the observed Bell value $\langle I \rangle$. Minimal resource requirements for an observed Bell violation $v$ beyond the local bound $L$ are derived in the text.}
 \label{Introduction_figure}
\end{figure}
In a second step, we show that this criterion also provides the minimal amount of coherence, discord, and entanglement needed for the violation of an inequality with any Bell-diagonal Bell operator, which is of particular interest for the case of two-qubit systems. As an application of our results, we present a closed expression for the maximal possible violation of the Clauser-Horne-Shimony-Holt (CHSH) inequality \cite{PhysRevLett.23.880} given some fixed amount of entanglement or purity and a given level of measurement incompatibility. This enables us to establish a surprising link between the incompatibility of quantum measurements and the minimal entanglement needed. More precisely, we show that highly incompatible projective measurements need, in some instances, a higher amount of entanglement in order to show some fixed CHSH nonlocality than less incompatible projective measurements. In other words, a smaller resource on the measurement side does \textit{not} require a higher resource on the state side, which is counterintuitive. An analogous result follows for the case of the two-setting linear steering inequality \cite{PhysRevA.80.032112}.\\
\indent \textit{Preliminaries.}\textemdash
In general, we are considering Hermitian Bell operators of the form
\begin{align}
I = \sum\limits_{a,b,x,y} c_{ab \vert xy} \ M_{a \vert x} \otimes M_{b \vert y},      
\end{align}
with real coefficients $ c_{ab \vert xy} $ and positive semidefinite operators $ M_{a \vert x} $, $ M_{b \vert y} $ with outcomes $ a,b $ and inputs $ x,y $ which form a positive operator valued measure  (POVM) such that $\sum_a M_{a \vert x} = \mathds{1} $ and $\sum_b M_{b \vert y} = \mathds{1} $. 
A Bell inequality is given by 
\begin{align}
\mathrm{Tr}(I \rho_{LHV}) \leq L,    
\end{align} 
with the (real) local bound $ L $ for all correlations obeying a so-called local hidden-variable model (LHV). This inequality may be violated by some entangled quantum states $ \rho $. We call states which violate (at least) one Bell inequality nonlocal. The achieved Bell value is denoted by $ \langle I \rangle = \mathrm{Tr}(I \rho) = L+v $ where $ v > 0 $ is the amount by which the bound $ L $ is violated. During the course of this Letter, we will often use the spectral decomposition of a quantum state $ \rho = \sum_i^d \lambda_i \vert \phi_i  \rangle \langle \phi_i \vert $ with $ \lambda_i \geq 0$ and $\sum_i^d \lambda_i = 1 $ and the Bell operator $ I = \sum_j^d \mu_j \vert \Psi_j \rangle \langle \Psi_j \vert$ with real eigenvalues $\mu_j$ where $d$ is the dimension of $\rho$. The sets $ \lbrace \vert \phi_i  \rangle \rbrace $, $ \lbrace \vert \Psi_j  \rangle \rbrace $ form orthonormal bases. We order (without loss of generality) the eigenvalues in descending order, i.e. $ \lambda_i \geq \lambda_s $ for $ i < s  $ and $ \mu_j \geq \mu_t $ for $ j < t  $. \\
\indent \textit{Main task.}\textemdash
 We want to quantify the minimal quantum resources of a state $\rho$ in order to achieve some given violation $v$ for a (fixed) Bell operator $I$. Thus, we want to minimize a general resource quantifier $R(\rho)$ such that $ \rho $ is consistent with the observed data in terms of the Bell expectation value $\langle I \rangle $, i.e. we want to find
\begin{align}
R^* = \min_{\rho} \lbrace R(\rho) \ \vert \  \langle I \rangle = \mathrm{Tr}(\rho I) = L+v \rbrace.   \label{main_opt}
\end{align} 
Optimizations of this form naturally occur in inference schemes based on entanglement witnesses \cite{Audenaert2006,Eisert2007,PhysRevLett.98.110502,Cavalcanti2006}. The important difference to the task we consider here is that nonlocality itself is also a resource. In the context of nonlocality this problem has only been addressed for the CHSH inequality \cite{PhysRevA.59.1799,Batle2011,PhysRevA.65.024304,Batle2001,PhysRevA.88.052105,PhysRevLett.89.170401} with the main focus on entanglement. \\ 
Let us specify what we mean by the term quantum resource without going into details, for more information see \cite{RevModPhys.91.025001}. In any resource theory, one first defines the states which are no resource, the so-called void-states (or free-states) which constitute the set $ \mathcal{V} $. Second, one defines the (maximal) set of operations $ \Lambda $ (free operations) that cannot turn a void state into a resource state. Finally, one has to find measures $ R $ which quantify the respective resource. E.g. in the resource theory of entanglement, the free states are the separable ones, the free operations are local operations and classical communication (LOCC) and a quantifier is the relative entropy of entanglement. \\
 \textit{Purity.}\textemdash Our main results are analytical results for the minimal purity of $\rho$ needed to achieve the Bell value $\langle I \rangle = L+v$  for a general Bell operator $I$ of any dimension $d$, any number of parties $n$, settings $k$, and outcomes $m$. We will employ two different but common purity measures. First, we consider a generalized robustness quantifier
 \begin{align}
G_R(\rho) := \min_{\tau}  \left\{ x \vert x \geq 0, \exists \ \text{a state} \ \tau, \dfrac{\rho+x \tau}{1+x} \in \mathcal{V}  \right\}, \label{gen_robustness_purity}
\end{align} 
where the set $ \mathcal{V} $ consists of the void states. Since $\tau$ can be any state, $ G_R(\rho) $ can be seen as general noise robustness of $\rho$ against a void set $ \mathcal{V} $ and can therefore be used to quantify a general resource $G$. For purity the set $ \mathcal{V} $ contains only the maximally mixed state $ \mathds{1}/d $. It was shown in \cite{Luo2019} that 
\begin{align}
P_R(\rho) = d \lambda_1(\rho)-1.
\end{align}
Therefore, minimizing $ P_R(\rho) $ reduces to minimizing $ \lambda_1(\rho) $. In order to show our main result we first answer the (easier to solve) reverse question: given $P_R(\rho)$ what is the maximal possible Bell value $\langle I \rangle_{\text{max}} = L + v_{\text{max}} $ the state $\rho$ can achieve for a fixed Bell operator $I$?
\begin{theorem}
\label{thrm1}
Given the Hermitian operator $ I = \sum_{j=1}^d \mu_j \vert \Psi_j \rangle \langle \Psi_j \vert$ with $ \mu_j\geq\mu_t $ for $j<t$ and a fixed robustness of purity $P_R(\rho)$ of a quantum state $\rho$. The maximal expectation value $\langle I \rangle_{\text{max}} $ can be achieved by  $\rho = \sum_{i=1}^r \lambda_i \vert \Psi_i \rangle \langle \Psi_i \vert $, where $ \lambda_i \geq 0$, $\sum_{i=1}^r \lambda_i = 1 $, $ \lambda_i \geq \lambda_s $ for $ i < s $, and is given by
\begin{align}
\langle I \rangle_{\text{max}} = \sum\limits_{j=1}^r  \mu_j \lambda_j, \label{max_violation}
\end{align}
where $r$ is an integer s.t. $ \dfrac{1}{r-1} > \lambda_1 \geq \dfrac{1}{r} $ and all eigenvalues $\lambda_i$ for $i \in \lbrace 1, \cdots r-1 \rbrace$ are equal to $\lambda_1 = (1+P_R)/d$.
\end{theorem}
\begin{proof}
The theorem follows from the generalization of Ruhe's trace inequality \cite{Marshall2011} and the fact that it is optimal to choose all eigenvalues $\lambda_i$ equal to $\lambda_1$ except the lowest non-zero one, which is given by normalization. The integer $r$ defines the rank of the optimal $\rho$ which we construct from the $\lbrace \lambda_i \rbrace$ and the eigenstates of $I$. This choice is unique for non-degenerate eigenvalues of $I$. See \cite{Supplemental_Material} for the specifics of the proof.
\end{proof}
Theorem. \ref{thrm1} can be used reversely (see Lemma \ref{lemma1} in the Supplemental Material \cite{Supplemental_Material}), which provides our first main result. Namely, for given $ \langle I \rangle_{\text{max}} $ we can use eq. (\ref{max_violation}) to determine the minimal $P_R(\rho)$ or $\lambda_1(\rho) $ needed to achieve the Bell value $ \langle I \rangle_{\text{max}} $. In order to determine $\lambda_1(\rho) $ one needs to find the integer $r$ s.t. Theorem \ref{thrm1} is valid. The usefulness of Theorem \ref{thrm1} lies in its simplicity. Not only does it allow to minimize the generalized robustness of purity $P_R(\rho)$ for a fixed expectation value of the most general Bell operator via an easily accessible criterion, also one needs to check at most $d$ linear equations.\\
Second, we show that $\langle I \rangle_{\text{max}} $ can also be determined analytically as a function of the most commonly used measure of purity, the Rényi $2$-purity $ \mathcal{P}_2 = \mathcal{P}_2(\rho) = \log_2{(d \mathrm{Tr(\rho^2)})} $ \cite{Purity_resource}.
\begin{theorem}
\label{thrm2}
Given the Hermitian operator $ I = \sum_{j=1}^d \mu_j \vert \Psi_j \rangle \langle \Psi_j \vert$ with $ \mu_j\geq\mu_t $ for $j<t$ and a fixed Rényi $2$-purity $\mathcal{P}_2(\rho)$ of a quantum state $\rho$. The maximal expectation value $\langle I \rangle_{\text{max}}$ can be achieved by $\rho = \sum_{i=1}^r \lambda_i \vert \Psi_i \rangle \langle \Psi_i \vert $, where $ \lambda_i \geq 0$, $\sum_{i=1}^r \lambda_i = 1 $, $ \lambda_i \geq \lambda_s $ for $ i < s $, and is given by
\begin{align} 
\langle I \rangle_{\text{max}} = \dfrac{G+\sqrt{(1-\dfrac{r}{d} 2^{\mathcal{P}_2} )(G^2-Hr)}}{r}, \label{Imax_tr_square}  
\end{align}
where $ G=\sum_i^r \mu_i $, $ H=\sum_i^r \mu_i^2 $, and $r \in \lbrace 1, \cdots, d \rbrace$ is the largest integer s.t. 
\begin{align}
\lambda_i = \dfrac{(r \langle I \rangle_{\text{max}}-G) \mu_i + H - G \langle I \rangle_{\text{max}} }{Hr-G^2} \geq 0 \ \forall \ i  \leq r.
\end{align}
\end{theorem}
\begin{proof}
The idea of the proof is to first find the spectrum $ \lbrace \lambda_i \rbrace $ of minimal purity $ \mathcal{P}_2 $ for a fixed expectation value $\langle I \rangle_{\text{max}} $, which we will do by invoking the method of Lagrange multipliers. This spectrum in turn defines the state $\rho$  maximizing the expectation value of the operator $I$ for fixed $\mathcal{P}_2$. The integer $r$ defines the rank of $\rho$. This choice is unique for non-degenerate eigenvalues of $I$. We relocate the specifics of this proof to the Supplemental Material \cite{Supplemental_Material}.
\end{proof}
While Theorem \ref{thrm2} has a more involved structure than Theorem \ref{thrm1}, some cases can still be easily seen. E.g. for pure states the maximum will be $ \langle I \rangle_{\text{max}} = \mu_1 $ while for the maximally mixed state $ \langle I \rangle_{\text{max}} = \mathrm{Tr}(I)/d $. \\
\textit{Equality of quantum resources for two qubits.}\textemdash In the following we show which effect minimizing the purity has on the other state resources. Due to its simplicity we focus on the generalized robustness.
In other words, we demonstrate the power of Theorem \ref{thrm1} by showing that for the subset of two-qubit correlation inequalities the states of minimal generalized robustness of purity for a fixed violation $v$ also minimize the respective generalized robustnesses of coherence $ C_R(\rho) $, discord $D_R(\rho)$, and entanglement $E_R(\rho)$, which in fact turn out to be equal. This is of particular interest since for every quantum state the hierarchy \cite{Purity_resource}
\begin{align}
\mathcal{P}(\rho) \geq \mathcal{C}(\rho) \geq \mathcal{D}(\rho) \geq \mathcal{E}(\rho), \label{hierarchy}  \end{align}
holds when quantified by the same distance-based \footnote{Note that the hierarchy also holds for measures which are no distance in a strict mathematical sense, like the relative entropy, due to the inclusion relations of the void-sets} measure and coherence is quantified with respect to any product basis. We will in particular choose the product basis that minimizes the coherence of the state $\rho$. This notion of coherence coincides with the notion of symmetric quantum discord with respect to all sub-systems \cite{Purity_resource,PhysRevA.92.022112}. Therefore, we will only summarize the concept of coherence \cite{RevModPhys_coherence} here, for more details about discord, see \cite{Bera2017}. Coherence in general, is a basis-dependent concept and is connected to the ability of a state to be in a superposition of some (fixed) basis states. The void states $\delta$ are called incoherent states. These are diagonal with respect to a fixed basis ${\vert i \rangle}$, i.e,
\begin{align}
\delta = \sum_i p_i \vert i \rangle \langle i \vert, \ p_i \geq 0, \ \sum_i p_i = 1. \end{align}
Note, that our notion of coherence corresponds to a minimization over all states equivalent to $\rho$ under local unitaries. \\
Our result is summarized in the following theorem.
\begin{theorem}
\label{thrm3}
Given a Bell operator of the form 
\begin{align}
I = \sum\limits_{x,y} g_{x,y} \ A_x \otimes B_y,    \label{correlation_inequality}    
\end{align}
with local observables $ A_x = \Vec{a}_x \cdot \Vec{\sigma} $, $ B_y = \Vec{b}_y \cdot \Vec{\sigma} $ where $ \Vec{a}_x, \Vec{b}_y $ are Bloch vectors and $ \Vec{\sigma} $ is the vector containing the Pauli matrices. For a fixed expectation value $\langle I \rangle = L + v$, where $L$ is the local bound and $v > 0$, there exists a two-qubit quantum state $\rho_{\text{opt}}$ which simultaneously minimizes the generalized robustnesses of purity $P_R$, coherence with respect to all product bases $C_R$, and entanglement $E_R$, s.t. $\rho_{\text{opt}}$ has an equal amount of all these resources (up to normalization).
\end{theorem}
\begin{proof}
The proof relies on the fact that the states of minimal entanglement are Bell-diagonal states (BDS), which are entangled iff $\lambda_1 > 1/2$ since the entanglement of BDS is a monotonic function of $\lambda_1$ only. Using this fact and Lemma \ref{lemma1} (see the Supplemental Material \cite{Supplemental_Material}) the optimal state $\rho_{\text{opt}}$ can always be chosen to be of at most rank 2. This enables us to show that the closest separable state is always incoherent in some product basis. Therefore minimizing $\lambda_1$ minimizes all state resources. We relocated the specifics of the proof to the Supplemental Material \cite{Supplemental_Material}
\end{proof} 
Note that an equivalence between coherence and entanglement for maximally correlated states has also been shown in different contexts \cite{Zhu2017,PhysRevLett.115.020403}.
We want to highlight that there is a straightforward generalization to genuine-multipartite entanglement (GME) quantification for $N-$qubit GHZ-diagonal Bell operators (e.g. two-setting full-correlation inequalities \cite{Scarani2001}) when we ask for a violation $v$ which requires GME \cite{PhysRevD.35.3066}, since the optimal states will then be diagonal in the GHZ basis and the GME of these states is completely characterized by $ \lambda_1>1/2 $ \cite{Cianciaruso2016}, analogously to two-qubit BDS. \\
However, in general the hierarchy (\ref{hierarchy}) will not be tight. Based on numerical optimization we find that there is indeed a (nontrivial) gap between purity, coherence, and entanglement for judiciously chosen observables in the $ I3322 $ inequality \cite{Collins2004}, an inequality with three settings and two-outcomes for both parties including single party correlators. See the Supplemental material \cite{Supplemental_Material} for more details.\\
\textit{CHSH inequality.}\textemdash  Remarkably, our results bring even new insights for the well-known CHSH inequality . The CHSH operator \cite{PhysRevLett.23.880} is defined as
\begin{align}
I = A_1 \otimes B_1 + A_1 \otimes B_2 + A_2 \otimes B_1 - A_2 \otimes B_2,     
\end{align}
with $ \vert  \langle I \rangle \vert \leq 2 $ for local-realistic models. The general form of eq. (\ref{max_violation}) can now be reduced to at most rank-$2$ solutions
\begin{align}
\lambda_1 \mu_1 + (1-\lambda_1) \mu_2 = L+v, \label{result_1}
\end{align}
which recovers the finding made in \cite{PhysRevA.65.024304}. Furthermore it is well known \cite{Landau1987} that if the observables fulfil $ A_i^2 = B_j^2 = \mathds{1} $, it holds
\begin{align}
I^2 = 4 \mathds{1} \otimes \mathds{1} - [A_1,A_2] \otimes [B_1,B_2],
\end{align}
where $ [X,Y] $ denotes the commutator between $X$ and $Y$. The observables $X$ and $Y$ describing projective measurements are called incompatible, i.e. they cannot be measured jointly iff $ [X,Y] \neq 0 $. This quantum effect is the central aspect of the famous Heisenberg-Robertson uncertainty relation \cite{PhysRev.34.163}. The use of incompatible measurements is necessary but not sufficient for Bell nonlocality \cite{PhysRevA.97.012129,Bene2018}. There exists a resource theory \cite{Incomop_resource} which allows the quantification of measurement incompatibility of one party.
Let us introduce as a quantifier $C$ for the (global) incompatibility the product of the single party measurement incompatibilities defined by the operator norm (largest absolute eigenvalue) of the commutators. Namely,
$ C = C_A C_B = \vert \vert [A_1,A_2] \vert \vert \cdot \vert \vert [B_1,B_2] \vert \vert $. This is well motivated since $C=0$ iff one of the parties holds compatible measurements, i.e. the CHSH inequality cannot be violated and $C=4$ is achieved with Pauli commutation relations only.
After some algebra, we obtain the eigenvalues of $ I $ as a function of $ C $, i.e. \cite{PhysRevLett.68.3259,Kiukas2010}
\begin{align}
\mu_{1/4} = \pm \sqrt{4+C},\ \mu_{2/3} = \pm \sqrt{4-C}. \label{CHSH_eigenvalues}
\end{align}
This shows that the quantity $C$ quantifies the maximal nonlocality which can possibly be revealed by the given observables. By introducing the global measurement incompatibility we can study relations between the necessary resources contained in the states and those contained in the measurements, when wanting to achieve a certain nonlocality. The maximal possible violation given in eq. (\ref{max_violation}) reduces to \begin{align}
\langle I \rangle_{\text{max}} = \sqrt{4+C}\lambda_1 + \sqrt{4-C} (1-\lambda_1).    \label{CHSH_violation_max}
\end{align}
Note that after inserting the optimal incompatibility $C_{\text{max}} = \dfrac{4(2\lambda_1 -1)}{2\lambda_1^2-2\lambda_1 +1} $ to maximize the Bell value $ \langle I \rangle_{\text{max}} $ for fixed $\lambda_1$ one easily recovers the special case \cite{Batle2011} and noteworthy the result \cite{PhysRevLett.89.170401} where a formula for the maximal CHSH value of a two-qubit state in terms of its concurrence was found. \\
Intuitively one would expect now for a fixed violation of the CHSH inequality, that there is a trade-off between the necessary measurement resources and the necessary state resources in the sense, that more of the resource in the measurements require less resource in the state. This, however, is not always the case. As one can see in Fig. \ref{CHSH_min_ent_vs_C} there are parameter-regions where less resources on the measurement side go together with less resources on the state side.  \begin{figure}[h!]
\includegraphics[scale =0.35]{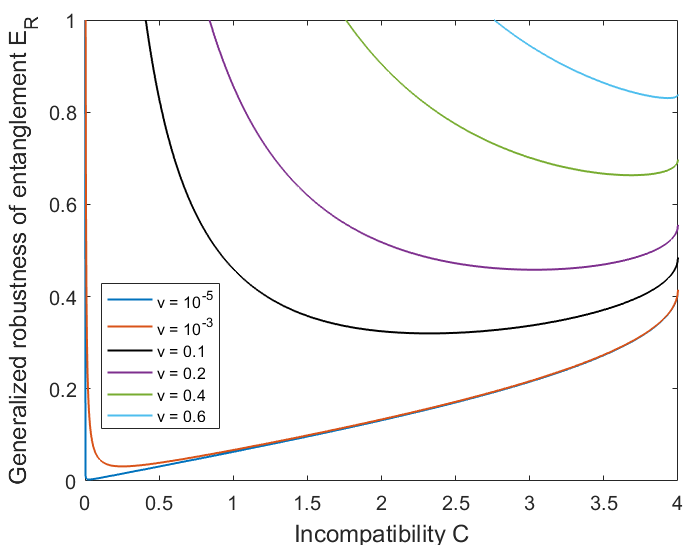}
 \caption{The minimal generalized robustness of entanglement $ E_R(\rho) $ for a given level of incompatibility $C$ for different amounts of desired violation $v$. The curves diverge at some $C$ because there is no state achieving the given violation. For low violations the effect that less entanglement for lower $C$ is necessary becomes clearly visible, for a large regime of $C$.}
 \label{CHSH_min_ent_vs_C}
\end{figure}
Especially for very small violations, weakly incompatible measurements require much less entanglement for the same amount of nonlocality.
Is the surprising behavior discussed above a generic feature, or does it possibly depend on the chosen quantifiers for measurement incompatibility and/or the resources in the state? We discuss other possible quantifiers for the incompatibility and the state resources in the Supplemental Material \cite{Supplemental_Material} and conclude that the behaviour is generic, by arguing that when other quantifiers are chosen only purity could possibly show a qualitatively different behavior. We show that this is indeed the case for the relative entropy of purity, while the Rényi $2$-purity shows a similar behaviour as the generalized robustness of purity as a function of $C$.\\
We strengthen this conclusion by highlighting that plots of the same qualitative behavior follow for the two-setting linear steering inequality \cite{PhysRevA.80.032112} given by
\begin{align}
F_2 = \vert \sum\limits_{i=1}^2 \langle A_i \otimes B_i \rangle \vert \leq \sqrt{2},
\end{align}
where Bobs measurements have to be aligned orthonormally while Alice is free to choose any projective measurements. In this case, the eigenvalues of the steering operator of $ F_2 $ only depend on $ C_A =  \vert \vert [A_1,A_2] \vert \vert $ in an analogous way to the CHSH inequality, i.e.
\begin{align}
\tilde{\mu}_{1/4} = \pm \sqrt{2+C_A},\ \tilde{\mu}_{2/3} = \pm \sqrt{2-C_A},
\end{align}
from which a behavior of the resources that is analogous to that for the CHSH inequality follows. This shows that the qualitative dependency of the state resources on the measurement incompatibility is not just due to our definition of the bipartite quantifier $C$, but a true physical phenomenon. \\
\textit{Discussion.}\textemdash In the present Letter we have shown that the minimal purity necessary to achieve a certain Bell value for the most general Bell operator can be found analytically via an easily accessible criterion. Since the purity of a state is its most fundamental resource which bounds \textit{all} other resources of this state, this has major consequences for the inference of other necessary resources like coherence and entanglement. We demonstrated this concretely by showing that the generalized robustness of all state resources can be minimized by the same state for two-qubit correlation inequalities. Finally, we have connected the nonlocality of quantum correlations, the incompatibility of quantum measurements and the state's resources via the CHSH inequality. This revealed the counterintuitive effect, that sometimes more state resources are required to reach the same level of nonlocality, when the measurement resources are increased. While the CHSH inequality is by far the most studied Bell inequality, this behaviour has, to the best of our knowledge, not been reported so far. The same effect is also prevalent for a steering inequality and thus excludes the existence of any possible conservation law for the necessary resources in states and measurements, regarding steering. \\ 
Our work leads to some immediate directions of future research. First, except for the case of purity which is now solved completely, one should investigate more general Bell scenarios. Second, one could investigate further important resource measures. Finally, one should further investigate how the spectrum of Bell operators depends on the physical properties of the used measurement operators. \\

We thank Lennart Bittel, Federico Grasselli, Martin Kliesch, and Gláucia Murta for helpful discussions.
This research was partially supported
by the EU H2020 QuantERA ERA-NET Cofund in
Quantum Technologies project QuICHE, and by the Federal Ministry of Education and Research (BMBF).

\bibliography{bibliography_2.bib}

\onecolumngrid
\appendix

\section{Supplemental Material for "Quantifying necessary quantum resources for nonlocality"}
\indent \textit{Maximal expectation value for fixed generalized robustness of purity $P_R(\rho)$.} \textemdash Here, we give more detailed derivation of Theorem (\ref{thrm1}) in the main text. 
\begin{theorem1a}
Given the Hermitian operator $ I = \sum_{j=1}^d \mu_j \vert \Psi_j \rangle \langle \Psi_j \vert$ with $ \mu_j\geq\mu_t $ for $j<t$ and a fixed robustness of purity $P_R(\rho)$ of a quantum state $\rho$. The maximal expectation value $\langle I \rangle_{\text{max}} $ can be achieved by  $\rho = \sum_{i=1}^r \lambda_i \vert \Psi_i \rangle \langle \Psi_i \vert $, where $ \lambda_i \geq 0$, $\sum_{i=1}^r \lambda_i = 1 $, $ \lambda_i \geq \lambda_s $ for $ i < s $, and is given by
\begin{align}
\langle I \rangle_{\text{max}} = \sum\limits_{j=1}^r  \mu_j \lambda_j,
\end{align}
where $r$ is an integer s.t. $ \dfrac{1}{r-1} > \lambda_1 \geq \dfrac{1}{r} $ and all eigenvalues $\lambda_i$ for $i \in \lbrace 1, \cdots r-1 \rbrace$ are equal to $\lambda_1 = (1+P_R)/d$.
\end{theorem1a}
\begin{proof}
It holds $ P_R(\rho) = d \lambda_1(\rho) -1 $, i.e. the constraint of fixed purity $ P_R(\rho) $ depends only on the largest eigenvalue of $\rho$ and we have to do the optimization w.r.t. the remaining degrees of freedom. The generalization of Ruhe's trace inequality \cite{Marshall2011} states that if $A,B$ are $d \times d$ Hermitian matrices, then
\begin{align}
\sum\limits_{i=1}^d \eta_i(A) \eta_{d-i+1}(B) \leq \mathrm{Tr}(AB) \leq \sum\limits_{i=1}^{d} \eta_i(A) \eta_i(B) \end{align}
where $\eta_i(A)$ denotes the $i$-th eigenvalue of A and the eigenvalues $\eta_i$ are ordered in descending order. This simply means that in order to maximize (minimize) the expectation value $\langle I \rangle = \mathrm{Tr}(\rho I)$, $\rho$ has to be diagonal in the same basis as $I$. Further, the $i$-th largest eigenvalue of $\rho$ has to be multiplied with the $i$-th largest (smallest) eigenvalue of $I$. Since $\lambda_1$ is by assumption the largest eigenvalue of $\rho$, all other $\lambda_i$ cannot be larger. In order to maximize the magnitude of the expectation value all $\lambda_i$ should be as big as possible, which means equal to $\lambda_1$. This however is only possible for $r-1$ eigenvalues, where $r$ (which describes the rank of $\rho$) is the largest integer such that the $\lbrace \lambda_i \rbrace$ describe a normalized quantum state. The remaining non-zero eigenvalue $\lambda_r$ is given by the normalization constraint which also provides the upper bound to $\lambda_1$. The lower bound comes from the requirement that $\lambda_1$ is the largest eigenvalue, which finishes the proof.
\end{proof}

As an extension, we show in the following Lemma that Theorem (\ref{thrm1}) can also be used to determine the minimal $\lambda_1$ for a fixed expectation value $\langle I \rangle_{max}$.

\begin{Lemma}
\label{lemma1}
Given the Hermitian operator $ I = \sum_{j=1}^d \mu_j \vert \Psi_j \rangle \langle \Psi_j \vert$ with $ \mu_j\geq\mu_t $ for $j<t$. The minimal robustness of purity $P_R(\rho)$ of a quantum state $ \rho = \sum_{i=1}^d \lambda_i \vert \phi_i \rangle \langle \phi_i \vert $, where $ \lambda_i \geq 0$ and $\sum_{i=1}^d \lambda_i = 1 $,  $ \lambda_i \geq \lambda_s $ for $ i < s $, achieving the expectation value $\langle I \rangle_{max} \geq \dfrac{1}{d} \mathrm{Tr}(I)$ is determined by the equation
\begin{align}
\langle I \rangle_{\text{max}} = \sum\limits_{j=1}^r  \mu_j \lambda_j, 
\end{align}
provided that $ \dfrac{1}{r-1} > \lambda_1 \geq \dfrac{1}{r} $, and all eigenvalues $\lambda_i$ for $i \in \lbrace 1, \cdots r-1 \rbrace$ are equal to $\lambda_1$.
\end{Lemma}

\begin{proof}
Since $ P_R(\rho) = d \lambda_1(\rho) -1 $, we need to minimize $\lambda_1$. The expectation value in general is given by 
\begin{align}
\langle I \rangle_{\text{max}} = \mathrm{Tr}(I \rho) = \sum_{i,j} \mu_i \lambda_j \vert \langle \Psi_i \vert \phi_j \rangle \vert^2. \label{expect_Bell}
\end{align}
We achieve a minimization of $\lambda_1$ by exploiting the following two observations.   
First, since $\langle I \rangle_{max} \geq \dfrac{1}{d} \mathrm{Tr}(I)$ \footnote{As $P_R(\mathds{1}/d) = 0$ and $\langle I \rangle_{\mathds{1}/d} = \mathrm{Tr}(I)/d$. For expectation values $\langle I \rangle < \mathrm{I}/d $ the proof works similar by simply exchanging the order of the $\mu_i$ from decreasing to ascending.} we need to maximize for fixed $j$ the term $ \sum_i \mu_i \lambda_j \vert \langle \Psi_i \vert \phi_j \rangle \vert^2 $, i.e. we need to appropriately choose the eigenbasis of $\rho$ which according to Theorem (\ref{thrm1}) will be done by choosing $\rho$ diagonal in the same basis as $I$, or more specifically $ \vert \phi_j \rangle = \vert \Psi_j \rangle \ \forall \ j$. This can be seen by realizing that $ \sum_i \mu_i \lambda_j \vert \langle \Psi_i \vert \phi_j \rangle \vert^2 $ is upper bounded by $\lambda_j \mu_j$. Hence, if $ \vert \phi_j \rangle \neq \vert \Psi_j \rangle$, $\lambda_1$ has to be larger than necessary because some part of the contribution towards the expectation value is lost due to a sub-optimal basis choice. Second, $\lambda_1$ will be minimal for maximal possible $\lambda_2,\lambda_3,\cdots, \lambda_r$, where $r$ is just an index for now. More precisely, it is optimal to choose as many $\lambda_i$ equal to $\lambda_1$ as possible, since by definition $\lambda_1$ is the maximal eigenvalue and for any lower value of the $\lambda_i$, $\lambda_1$ would again be larger than necessary since we did not choose the maximal contribution of the terms $\lambda_i \mu_i \ \forall \ i$ towards the expectation value. However, we still have to incorporate that $\rho$ is a normalized quantum state, which means not all $\lambda_i$ can actually be equal to $\lambda_1$. In general it possible to choose all $\lambda_i$ equal to $\lambda_1$ for $i \in \lbrace 1, \cdots, r-1 \rbrace$ and to determine the smallest non-zero eigenvalue $\lambda_r$ by normalization. Hence $r$ denotes the rank of $\rho$. This leads to
\begin{align}
\langle I \rangle_{\text{max}} = \sum\limits_{j=1}^r  \mu_j \lambda_j, \label{expect_Bell2}
\end{align}
where we still do not know the value of the rank $r$. However, we can just make an Ansatz for some $r \in \lbrace 1,\cdots d´\rbrace$ and check whether the conditions $ \dfrac{1}{r-1} > \lambda_1 \geq \dfrac{1}{r} $, which are necessary for $\rho$ to be normalized, positive-semidefinite, and for $\lambda_1$ to be the largest eigenvalue of $\rho$ are fulfilled or not.
\end{proof}

\indent \textit{Maximal expectation value for fixed  Rényi $2$-purity.}\textemdash Here, we give a detailed derivation of Theorem (\ref{thrm2}) in the main text.
\begin{theorem2}
Given the Hermitian operator $ I = \sum_{j=1}^d \mu_j \vert \Psi_j \rangle \langle \Psi_j \vert$ with $ \mu_j\geq\mu_t $ for $j<t$ and a fixed Rényi $2$-purity $\mathcal{P}_2(\rho)$ of a quantum state $\rho$. The maximal expectation value $\langle I \rangle_{\text{max}}$ can be achieved by $\rho = \sum_{i=1}^r \lambda_i \vert \Psi_i \rangle \langle \Psi_i \vert $, where $ \lambda_i \geq 0$, $\sum_{i=1}^r \lambda_i = 1 $, $ \lambda_i \geq \lambda_s $ for $ i < s $, and is given by
\begin{align} 
\langle I \rangle_{\text{max}} = \dfrac{G+\sqrt{(1-\dfrac{r}{d} 2^{\mathcal{P}_2} )(G^2-Hr)}}{r},
\end{align}
where $ G=\sum_i^r \mu_i $, $ H=\sum_i^r \mu_i^2 $, and $r \in \lbrace 1, \cdots, d \rbrace$ is the largest integer s.t. 
\begin{align}
\lambda_i = \dfrac{(r \langle I \rangle_{\text{max}}-G) \mu_i + H - G \langle I \rangle_{\text{max}} }{Hr-G^2} \geq 0 \ \forall \ i  \leq r.
\end{align}
\end{theorem2}
\begin{proof}
 We first show that we can solve a different but connected optimization task which will lead to a proof of the theorem. The first simplification will be, that instead of considering the Rényi $2$-purity $ \mathcal{P}_2 $ directly, we can just consider the linear purity $ P = \mathrm{Tr(\rho^2)} $, since the logarithm is a monotonic function of the linear purity. Note that for any given purity $P$ the maximal achievable expectation value is s.t. $\langle I \rangle_{\text{max}} \geq \dfrac{1}{d} \mathrm{Tr}(I)$. This follows from the fact that $\langle I \rangle = \dfrac{1}{d} \mathrm{Tr}(I)$ is the expectation value achieved by the maximally mixed state (which has zero purity) and for any other purity $P >0$ one is able to choose a state $ \Omega = x \vert \Psi_1 \rangle \langle \Psi_1 \vert + (1-x) \dfrac{\mathds{1}}{d} $ with appropriately chosen $x \in [0,1]$ s.t. the purity of $\Omega$ is $P$. Since $ \vert \Psi_1 \rangle \langle \Psi_1 \vert $ is the eigenstate corresponding to the largest eigenvalue of $I$ it follows trivially that $\langle I \rangle_{\Omega} \geq \dfrac{1}{d} \mathrm{Tr}(I)$. This allows us to formulate an alternative optimization problem which proves the theorem. Given the Hermitian operator $I = \sum_j \mu_j \vert \Psi_j \rangle \langle \Psi_j \vert$ with fixed expectation value $\langle I'\rangle$ we want to find the minimal linear purity $P^*$ of a valid quantum state that achieves the expectation value $\langle I'\rangle$. We will now show by contradiction that for the minimal purity $P^*$ it holds $\langle I'\rangle = \langle I \rangle_{\text{max}}$. First, it is trivial that  $\langle I'\rangle > \langle I \rangle_{\text{max}}$ leads to a contradiction since $\langle I \rangle_{\text{max}}$ is by assumption the maximal expectation value for a given purity $P^*$. Second, if $\langle I'\rangle < \langle I \rangle_{\text{max}}$ we could construct a state $ \tilde{\rho} = t \rho_{\text{max}} + (1-t) \dfrac{\mathds{1}}{d} $ where $ \rho_{\text{max}} $ is a state of purity $P^*$ achieving the expectation value $ \langle I \rangle_{\text{max}} $ and choose $t \in (0,1) $ s.t. $\langle I \rangle_{\tilde{\rho}} = \langle I' \rangle$. It follows now for the purity of $\tilde{\rho}$ that 
\begin{align}
P(\tilde{\rho}) = P(t\rho_{\text{max}}+(1-t)\dfrac{\mathds{1}}{d}) < tP(\rho_{\text{max}}) + (1-t)P(\dfrac{\mathds{1}}{d}) \leq P^*,   
\end{align}
where we used the strict convexity of the linear purity (which is the square of the Frobenius norm) and the fact that $ P(\rho_{\text{max}}) = P^* $. This however is a contradiction, since $P^*$ is by assumption the minimal purity for states achieving the expectation value $\langle I'\rangle$ and we showed that $\tilde{\rho}$ would have a smaller purity while achieving the expectation value $\langle I'\rangle$. This allows us to solve the minimization problem for a fixed expectation value and use the optimal state of purity $P$ to solve the problem we consider in the theorem.
With the same argumentation as for the generalized robustness in the main text or alternatively the proof shown in \cite{PhysRevLett.89.170401} (originally in the context of the CHSH inequality) we can reduce the problem of minimizing $P$ under the Bell constraint s.t. the Bell operator $  I = \sum_i \mu_i \vert \Psi_i  \rangle \langle \Psi_i \vert  $ and the optimal quantum state  $\rho_{\text{opt}} = \sum_i \lambda_i \vert \Psi_i  \rangle \langle \Psi_i \vert $ will be diagonal in the same basis. Note that the eigenvalues of both operators are ordered in descending order i.e., $ \lambda_i \geq \lambda_s $ for $ i < s  $ and $ \mu_j \geq \mu_t $ for $ j < t  $. The Lagrangian of the problem is given by
\begin{align}
\mathcal{L}(\lambda_i,\alpha,\beta) = \sum_i \lambda_i^2 - \alpha (\sum_i \lambda_i -1) - \beta (\sum_i \lambda_i \mu_i - \langle I \rangle),
\end{align}
where $ \alpha $ is the multiplier according to the normalization constraint and $ \beta $ the multiplier for the expectation value constraint. We ignored for now the non-negativity of the eigenvalues but will come back to it later. Note that due to the eigenvalue ordering $\langle I \rangle \geq \dfrac{1}{d}\sum\limits_i \mu_i = \dfrac{1}{d} \mathrm{Tr}(I)$, since $\lambda_i \mu_i > \lambda_s \mu_s $ for $i<s$. Obviously this does not represent a loss of generality for our theorem since we are only interested in values $ \langle I \rangle \geq \dfrac{1}{d}\sum\limits_i \mu_i = \dfrac{1}{d} \mathrm{Tr}(I) $ anyway, as shown above. To find an optimum, we have to take the partial derivatives of $ \mathcal{L}(\lambda_i,\alpha,\beta)  $ with respect to the eigenvalues and the Lagrange multiplier. In the case of the multipliers we simply retrieve the constraints, for the eigenvalues we find 
\begin{align}
\dfrac{\partial \mathcal{L}}{\partial \lambda_k} = 2\lambda_k - \alpha - \beta \mu_k \overset{!}{=} 0,\end{align} \label{lagrange_condition}
which results in 
\begin{align}
\lambda_k = \dfrac{1}{2}(\beta \mu_k + \alpha).   \label{Lagrange_opt} 
\end{align}
From the normalization constraint $ \sum_k \lambda_k = 1 $ we get
\begin{align}
\alpha = \dfrac{2-\beta \sum\limits_{k=1}^r \mu_k}{r}, \end{align}
where the sum runs from $ k = 1 $ until the rank $r$, which means we make an Ansatz for a rank $r$ solution of the problem, i.e. all eigenvalues $  \lambda_i $ of the state are zero $ \forall \ i > r $. This method is not restrictive, in the sense that we are still able to find the true optimizer of the problem including the positivity constraints. The purity $ P = \mathrm{Tr}(\rho^2) $ is strictly convex, as it is the square of the Frobenius norm, which means there is unique global minimum. The projection (in the Frobenius norm) of a quasi-state $\eta$ with negative eigenvalues onto the feasible set of proper quantum states $\rho$ will be of smaller rank than $\eta$ itself. Therefore, we can start by making an Ansatz for an optimal state $ \rho_{\text{opt}} $ of full rank $ r = d $. If this is a proper quantum state according to (\ref{Lagrange_opt}), it is the state of minimal purity consistent with the expectation value constraint. If the state is not a proper quantum state, we consider solutions of rank $ r = d-1 $. By iteratively decreasing the rank, we will find a solution $ \rho_{\text{opt}} $ which is a proper quantum state and we do not have to consider states with an even lower rank. In the following, we will introduce the quantities  $ G = \sum_k^r \mu_k $ and $ H = \sum_k^r \mu_k^2 $. Using the Bell value constraint we find
\begin{align}
\beta (Hr-G^2) = 2\langle I \rangle r-2G, \label{beta}
\end{align}
which defines $\beta$ provided that $ Hr-G^2 \neq 0 $. It is easy to see that $ Hr-G^2 = 0 $ iff $\mu_i = \mu_1 \ \forall \ i \leq r$ i.e. the operator has $r$ degenerate largest eigenvalues $\mu_1$. This is a consequence of the Cauchy-Schwarz inequality when we consider the inner product of a vector $\Vec{\mu}$ containing the $r$ eigenvalues $\mu_i$ and a vector $\Vec{\mathds{1}}$ containing only ones. The case $ Hr-G^2 = 0 $ can however only lead to the optimal state $\rho_{\text{opt}}$ iff $\langle I \rangle = \mu_1$ which follows from the requirement that the r.h.s. of eq. (\ref{beta}) also vanishes. In this case the problem is independent of $\beta$ and we find $\lambda_k = \dfrac{1}{r}\ \forall \ k \leq r$.
If $ Hr-G^2 \neq 0 $ we find the optimal state $ \rho_{\text{opt}} $ by taking the eigenvalues from eq. (\ref{Lagrange_opt}) with the corresponding eigenstates from the operator $I$.
Reversely, this spectrum of eigenvalues can now be used to find the maximum expectation value $ \langle I \rangle_{\text{max}} $ for a fixed purity $P$. Multiplying (\ref{Lagrange_opt}) with $ \lambda_k $ and summing over $k$ leads to $P = \dfrac{1}{2}(\beta \langle I \rangle + \alpha)$, where we introduced the linear purity $ P = \mathrm{Tr}(\rho^2) = \sum_{k=1}^r \lambda_k^2 $. By solving for $ \langle I \rangle = \langle I \rangle_{\text{max}} = \sum_{k=1}^r \mu_k \lambda_k $ we will find 
\begin{align}
\langle I \rangle_{\text{max}} = \dfrac{G+\sqrt{(1-Pr)(G^2-Hr)}}{r}.    \label{Tr_2_max_suppl}
\end{align}
From eq. (\ref{Tr_2_max_suppl}) we see that for $P \geq 1/r_{\text{deg}}$ where $r_{\text{deg}}$ is the number of degenerate largest eigenvalues of $I$ the maximum is given by $ \mu_1 $, since $ G^2-Hr_{\text{deg}} = 0 $  which can trivially be checked to be correct. While for $ P \leq 1/r_{\text{deg}} $ for only $r_{\text{deg}}$ times degenerate eigenvalues $\mu_1$ of $I$, we see that $ G^2-Hr \neq 0 $ which shows that eq. (\ref{Tr_2_max_suppl}) is valid for all operators $I$. The theorem follows now by rewriting the purity $P$ in terms of the Rényi $2$-purity $ \mathcal{P}_2 $ and the requirement that all $\lambda_k \ \forall \ k = 1,\cdots,r$ are positive semidefinite (the remaining $d-r$ eigenvalues are zero) and inserting the Lagrange multipliers into the expression (\ref{Lagrange_opt}). Note that the optimal $ r $ for a given purity $ P $ can be found by the above described feasibility check. In other words, we check whether (\ref{Lagrange_opt}) leads to a valid quantum state for the rank $r$, purity $P$ and maximal expectation value $ \langle I \rangle_{\text{max}} $. Note further that in order to derive (\ref{Tr_2_max_suppl}), we have to solve a quadratic equation, which generally has two solutions. However, we are only interested in the maximum of those two solutions which is the solution corresponding to taking the positive root in (\ref{Tr_2_max_suppl}).
\end{proof}
\indent \textit{Equality of quantum resources for two qubits.}\textemdash Here, we give a detailed derivation of Theorem \ref{thrm3} in the main text. 
\begin{theorem3}
Given a Bell operator of the form 
\begin{align}
I = \sum\limits_{x,y} g_{x,y} \ A_x \otimes B_y,    \label{correlation_inequality_supp}    
\end{align}
with local observables $ A_x = \Vec{a}_x \cdot \Vec{\sigma} $, $ B_y = \Vec{b}_y \cdot \Vec{\sigma} $ where $ \Vec{a}_x, \Vec{b}_y $ are Bloch vectors and $ \Vec{\sigma} $ is the vector containing the Pauli matrices. For a fixed expectation value $\langle I \rangle = L + v$, where $L$ is the local bound and $v > 0$, there exists a two-qubit quantum state $\rho_{\text{opt}}$ which simultaneously minimizes the generalized robustnesses of purity $P_R$, coherence with respect to all product bases $C_R$, and entanglement $E_R$, s.t. $\rho_{\text{opt}}$ has an equal amount of all these resources (up to normalization).
\end{theorem3}
\begin{proof}
Let us first note that Bell operators of the form (\ref{correlation_inequality_supp}) are (up to local unitaries) diagonal in the Bell basis \cite{PhysRevLett.89.170401}. This means $ I = \sum_j \mu_j \vert \Psi_j  \rangle \langle \Psi_j \vert $ where the $ \lbrace  \vert \Psi_j  \rangle \rbrace $ are maximally entangled states. As a consequence, the states of minimal entanglement (for all measures) $ \rho_{BDS} = \sum_i \lambda_i \vert \Psi_j  \rangle \langle \Psi_j \vert $ are also diagonal in this Bell basis \cite{PhysRevA.59.1799}. 
 These states are entangled iff $ \lambda_1(\rho_{BDS}) > 1/2 $ and it was shown in \cite{Zhu2010} that  the generalized robustness of entanglement is given by $ E_R(\rho) = 2 \lambda_1(\rho) -1 $. This means minimizing $E_R(\rho)$ is equivalent to minimizing the generalized robustness of purity $P_R(\rho) = d \lambda_1(\rho) -1$, since both are monotonic functions of $\lambda_1(\rho)$. Due to Lemma (\ref{lemma1}) and the fact that $\lambda_1(\rho) > 1/2$, the state $\rho_{opt}$ can always chosen to be of the form  $ \rho_{\text{opt}} = \lambda_1 \vert \Psi_1 \rangle \langle \Psi_1 \vert + (1-\lambda_1) \vert \Psi_2 \rangle \langle \Psi_2 \vert $ where $\lambda_1$ can be determined from eq. (\ref{expect_Bell2}). It is convenient to choose  $ \tau = \vert \Psi_2 \rangle \langle \Psi_2 \vert $ as optimal noisy state in eq.(\ref{gen_robustness_purity}) to minimize the generalized robustness of entanglement. This is always possible since mixing the minimal noise necessary, i.e. $x = 2 \lambda_1-1$ of $\tau$ with $\rho_{\text{opt}}$ will end up in a separable state
 \begin{align}
 \dfrac{\rho_{\text{opt}}+(2\lambda_1-1) \tau}{1+(2 \lambda_1-1)} = \dfrac{1}{2}(\vert \Psi_1 \rangle \langle \Psi_1 \vert + \vert \Psi_2 \rangle \langle \Psi_2 \vert) \eqqcolon \xi. \label{closest_sep}  
 \end{align}
 Since all considered quantifiers are invariant under local unitaries, we fix the Bell operator  w.l.o.g. to be of the form
 \begin{align}
I = &\mu_1 \vert \Phi^+ \rangle \langle \Phi^+ \vert + \mu_2 \vert \Phi^- \rangle \langle \Phi^- \vert \label{Bell_diagonal_I}\\
+&\mu_3 \vert \Psi^+ \rangle \langle \Psi^+ \vert + \mu_4 \vert \Psi^- \rangle \langle \Psi^- \vert, \nonumber 
\end{align}
where $ \vert \Phi^{\pm} \rangle = \dfrac{1}{\sqrt{2}}( \vert 00 \rangle \pm \vert 11 \rangle)$ and $ \vert \Psi^{\pm} \rangle = \dfrac{1}{\sqrt{2}}( \vert 01 \rangle \pm \vert 10 \rangle) $ or any permutation of the eigenvalues. It follows now directly from the form of the closest separable state $\xi$ in eq. (\ref{closest_sep})
that it is also incoherent in some product basis, which means it is also the closest incoherent state to $\rho_{\text{opt}}$. This is because the hierarchy (\ref{hierarchy}) has to hold, which means $C_R(\rho_{\text{opt}}) \geq E_R(\rho_{\text{opt}})$. It can easily be seen that $\xi$ is incoherent for the mixtures of $ \lbrace \vert \Phi^+ \rangle \langle \Phi^+ \vert, \vert \Phi^- \rangle \langle \Phi^- \vert \rbrace $ or $ \lbrace \vert \Psi^+ \rangle \langle \Psi^+ \vert, \vert \Psi^- \rangle \langle \Psi^- \vert \rbrace $ and the computational basis. For the other combinations one finds as optimal bases tensor products of the eigenstates of the Pauli matrices $ \sigma_x $ or $ \sigma_y $.  The eigenstates of the the Pauli matrix $ \sigma_z $ are given by $ \lbrace \vert 0 \rangle, \vert 1 \rangle \rbrace, $ the eigenstates of $ \sigma_x $ by $ \vert \pm \rangle = \dfrac{1}{\sqrt{2}}(\vert 0 \rangle \pm \vert 1 \rangle ) $ and those of $ \sigma_y $ are given by $ \vert R/L \rangle  = \dfrac{1}{\sqrt{2}}( \vert 0 \rangle \pm i \vert 1 \rangle )$. Below we list for any of the possible rank-$2$ combinations the specific decomposition of the closest separable state $\xi$ into these states. Namely, it can easily be verified that
\begin{align}
&\lbrace \vert \Phi^+ \rangle \langle \Phi^+ \vert, \vert \Phi^- \rangle \langle \Phi^- \vert \rbrace \Rightarrow \xi = \dfrac{1}{2}(\vert 00 \rangle \langle 00 \vert + \vert 11 \rangle \langle 11 \vert),  \label{sigma_suppl}  \\
&\lbrace \vert \Phi^+ \rangle \langle \Phi^+ \vert, \vert \Psi^+ \rangle \langle \Psi^+ \vert \rbrace \Rightarrow \xi = \dfrac{1}{2}(\vert ++ \rangle \langle ++ \vert + \vert -- \rangle \langle -- \vert), \nonumber \\
&\lbrace \vert \Phi^+ \rangle \langle \Phi^+ \vert, \vert \Psi^- \rangle \langle \Psi^- \vert \rbrace \Rightarrow \xi = \dfrac{1}{2}(\vert RL \rangle \langle RL \vert + \vert LR \rangle \langle LR \vert), \nonumber \\
&\lbrace \vert \Phi^- \rangle \langle \Phi^- \vert, \vert \Psi^+ \rangle \langle \Psi^+ \vert \rbrace \Rightarrow \xi = \dfrac{1}{2}(\vert RR \rangle \langle RR \vert + \vert LL \rangle \langle LL \vert), \nonumber \\
&\lbrace \vert \Phi^- \rangle \langle \Phi^- \vert, \vert \Psi^- \rangle \langle \Psi^- \vert \rbrace \Rightarrow \xi = \dfrac{1}{2}(\vert -+ \rangle \langle -+ \vert + \vert +- \rangle \langle +- \vert), \nonumber \\
&\lbrace \vert \Psi^+ \rangle \langle \Psi^+ \vert, \vert \Psi^- \rangle \langle \Psi^- \vert \rbrace \Rightarrow \xi = \dfrac{1}{2}(\vert 01 \rangle \langle 01 \vert + \vert 10 \rangle \langle 10 \vert), \nonumber 
\end{align}
all other cases of rank-$2$ BDS are equivalent under local unitaries to one of the above cases.
This finishes the proof.
\end{proof} 
\indent \textit{The hierarchy of quantum resources in the context of the I3322 inequality.}\textemdash Here, we show that the hierarchy 
\begin{align}
P_R \geq C_R \geq D_R \geq E_R,  
\end{align}
is not always tight for the states minimizing the respective generalized robustnesses for a given Bell operator $I$. The I3322 inequality is given by \cite{Collins2004}
\begin{align}
&\langle A_1 \rangle + \langle A_2 \rangle - \langle B_1 \rangle - \langle B_2 \rangle  \\
+&\langle A_1 B_1 \rangle + \langle A_2 B_1 \rangle + \langle A_3 B_1 \rangle + \langle A_1 B_2 \rangle \nonumber \\
+&\langle A_2 B_2 \rangle - \langle A_3 B_2 \rangle + \langle A_1 B_3 \rangle - \langle A_2 B_3 \rangle \leq 4. \nonumber 
\end{align}
By generating random projective measurements $ \lbrace A_{a \vert x} \rbrace $ (and similar for Bob) which lead to the observables $ A_x = A_{2 \vert x} - A_{1 \vert x} $ we searched for Bell operators $ I $, which will lead to a non-tight hierarchy. We found that the following measurements of Alice and Bob do lead to such a case. The measurements (rounded to four digits) are given in the computational basis $ \lbrace \vert 0 \rangle, \vert 1 \rangle \rbrace $ by \begin{align}
A_{1 \vert 1} = \begin{pmatrix}
0.4379 & \phantom{-}0.3455 + 0.3560i \\
\phantom{-}0.3455 - 0.3560i & 0.5621 \\
\end{pmatrix}, \\
A_{1 \vert 2} = \begin{pmatrix}
0.6885 & \phantom{-}0.3964 - 0.2394i \\
\phantom{-}0.3964 + 0.2394i & 0.3115 \\
\end{pmatrix} \nonumber,\\
A_{1 \vert 3} = \begin{pmatrix}
0.9187 & -0.0737 +0.2632i \\
-0.0737 - 0.2632i & 0.0813 \\
\end{pmatrix} \nonumber, \\
B_{1 \vert 1} = \begin{pmatrix}
0.6973 & \phantom{-}0.0630 - 0.4551i \\
\phantom{-}0.0630 + 0.4551i & 0.3027 \\
\end{pmatrix} \nonumber, \\
B_{1 \vert 2} = \begin{pmatrix}
0.8982 & -0.2538 + 0.1645i \\
-0.2538 - 0.1645i & 0.1018 \\
\end{pmatrix} \nonumber, \\
B_{1 \vert 3} = \begin{pmatrix}
0.6472 & -0.0110 + 0.4777i \\
-0.0110 - 0.4777i & 0.3528 \\
\end{pmatrix} \nonumber, 
\end{align}
where the remaining POVM-elements are obtained by the completeness relation $\sum_a A_{a \vert x} = \mathds{1} $ (and similar for Bob). For the generalized robustnesses of purity, coherence, and entanglement we find for these settings and a required Bell value of $ \langle I \rangle = 4.001 $ that,
\begin{align}
P_R > C_R > E_R,   
\end{align}
with $ P_R = 2.6756 $,  $ C_R = 0.8418 $, and $ E_R = 0.8291 $. We calculated the purity robustness analytically with the methods described in the main text. The entanglement robustness was determined by semidefinite programming \cite{boyd_vandenberghe_2004}, and for the coherence robustness we used a combination of a simplex algorithm and semidefinite programming over all Bell operators of the form
\begin{align}
\tilde{I} = (U_A \otimes U_B) I (U_A \otimes U_B)^{\dagger},  
\end{align}
where $ U_A, U_B $ are local unitaries. For the simplex algorithm we used different randomly initialized starting points which all lead to the same result, suggesting it is the true minimum. Note that the gap between purity and coherence is not just due to a trivial factor like for correlation inequalities (see main text) but due to entirely different optimal states. \\
\indent \textit{Discussion on the definition of the quantifier $C$.}\textemdash In the main text we defined the quantifier $ C = \vert \vert [A_1,A_2] \vert \vert \cdot \vert \vert [B_1,B_2] \vert \vert  $ in order to judge the quality of the observables in the CHSH scenario. While the magnitude of a single party's commutator is a valid incompatibility monotone for projective measurements, it is not clear that $ C $ has a similar meaning for the measurements of both parties. As we argued in the main text, $ C $ is meaningful for the CHSH inequality since it determines the eigenvalues of the Bell operator and especially the maximal possible violation enabled by the observables $ \lbrace A_1,A_2,B_1,B_2 \rbrace $. In more general Bell scenarios, the ability to show nonlocality with some measurements is a distinct resource from measurement incompatibility \cite{PhysRevA.97.012129,Bene2018}. However, there exists so far no resource theory or straightforward quantification for the ability of observables to show nonlocality. In Fig. \ref{min_resource_CHSH} we show that indeed, the general resource requirement for higher Bell values increase.
\begin{figure}[h!]
\includegraphics[scale =0.35]{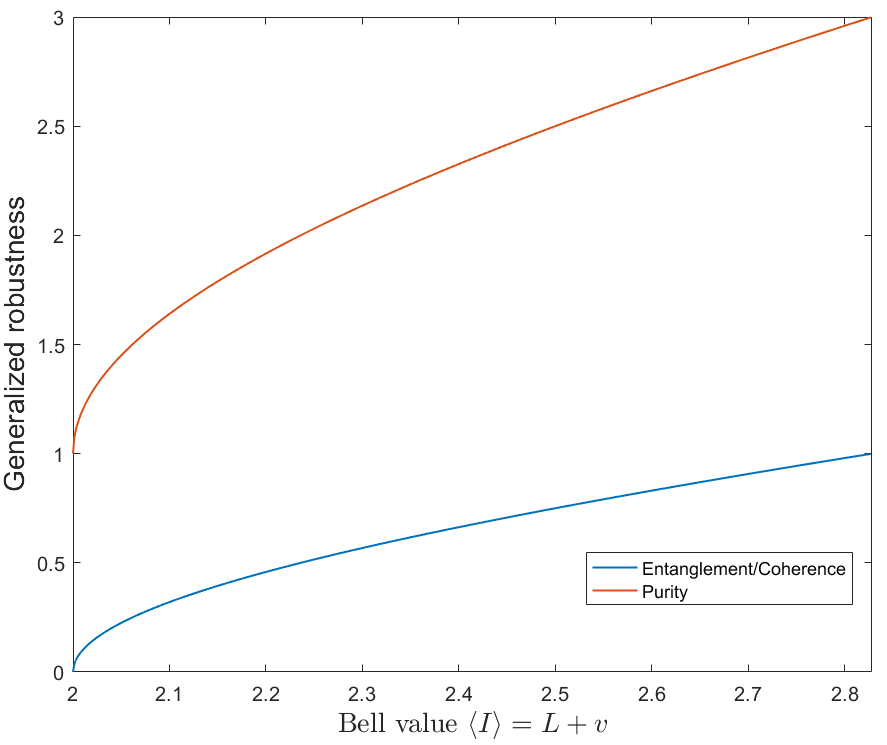}
 \caption{The minimal resources in terms of generalized robustness needed to achieve a fixed Bell value. As stated in the main text, coherence and entanglement are equivalent (and discord as well) while purity is the largest of these resources.}
 \label{min_resource_CHSH}
\end{figure} 
If we consider the sum of the single party's incompatibilities $ \tilde{C} = C_A + C_B = \vert \vert [A_1,A_2] \vert \vert  + \vert \vert [B_1,B_2] \vert \vert $, which due to linearity might be more intuitive to do, it would lead to eigenvalues 
\begin{align}
\mu_{1/4} = \pm \sqrt{4+\dfrac{(C_A+C_B)^2-C_A^2-C_B^2}{2}},\\
\mu_{2/3} = \pm \sqrt{4-\dfrac{(C_A+C_B)^2-C_A^2-C_B^2}{2}}. \nonumber
\end{align}
This however, is a problem since for fixed values of $ \tilde{C} $ there are multiple possible eigenvalue distributions $ \lbrace \mu_i \rbrace $ as one can easily see by fixing $ \tilde{C} = 2 $ and comparing the cases $ C_A = 2, C_B = 0$ and $ C_A = 1, C_B = 1 $. This means there cannot be a well defined function of the quantum resources of the state depending $ \tilde{C} $ in general. Note that in the special case $ C_A = C_B $, the eigenvalues will again be well defined. Finally, we could also be interested in the functional form of the resources depending on $ C_A $ and $C_B$. In this case, the fundamental effect that maximal incompatible measurements for Alice and Bob will not lead to the minimal required state resources will remain. Especially for very small violations $ v = 10^{-3} $, we show in Fig. \ref{CA_CB_incomp} that the behavior is essentially the same. 
\begin{figure}[h!]
\includegraphics[scale =0.35]{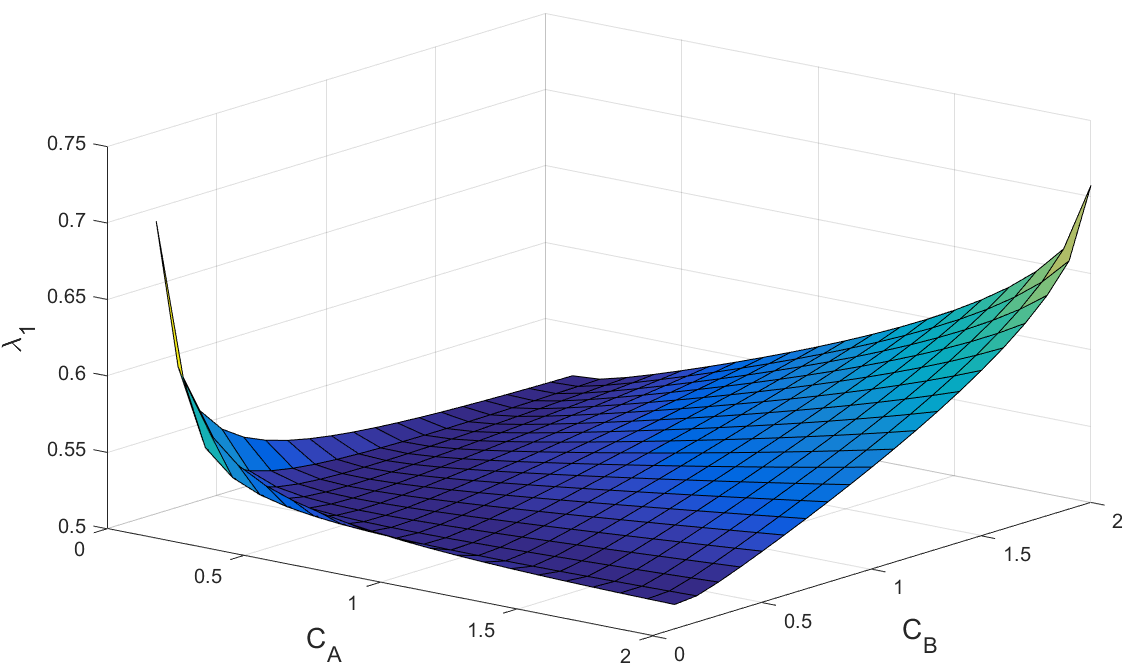}
 \caption{The necessary $ \lambda_1 $ vs. the single party's incompatibility quantifier $ C_A $, $C_B$ for a violation of $ v = 10^{-3} $. The plot shows that higher incompatibility resources need higher entanglement to achieve the same violation for some parameter region of $C_A, C_B$.}
 \label{CA_CB_incomp}
\end{figure} \\
\indent \textit{Variation of resource quantifiers.} Here, we discuss the qualitative influence when we chose a quantifier for the state's resources different than the generalized robustness. It turns out that only purity can show a qualitatively different behaviour as a function of $C$. This is the case for the relative entropy of purity. \\
First, since the entanglement of Bell diagonal states 
\begin{align}
\rho_{BDS} = \sum_i^4 \lambda_i \vert \Psi_j  \rangle \langle \Psi_j \vert, \end{align}
where the $ \vert \Psi_j  \rangle $ span an orthonormal basis of maximally entangled states, is completely characterized by their largest eigenvalue $\lambda_1$ and the resource measures are by definition resource monotones, it does not matter which distance-based entanglement quantifier we choose, since all of them will be monotonic functions of $\lambda_1$. By the same argumentation we can see, that for any rank-$2$ BDS state, the closest separable state $ \xi $ (see eq. (\ref{closest_sep})) does not change when we change the quantifier. Now the proof of Theorem (\ref{thrm3}) guarantees, that also the minimal coherence only depends on $\lambda_1$ and not on the concrete quantifier, since  the closest separable state $ \xi $ is incoherent in some product basis. However, if we chose a different purity measure, it will in general depend on the whole spectrum of the state and not just the largest eigenvalue. This can lead to a potentially different functional behaviour with respect to the incompatibility quantifier $C$. We show that this is indeed the case by considering the relative entropy of purity.
The relative entropy is defined as
\begin{align}
S(\rho \vert \vert \xi) = \min_{\xi \in \mathcal{V}} \mathrm{Tr}(\rho \log_2 \rho) - \mathrm{Tr}(\rho \log_2 \xi).    
\end{align}
In the case of purity, the relative entropy reduces to
\begin{align}
S_P(\rho \vert \vert \mathds{1}/d) = \log{d}-S(\rho), 
\end{align}
where $ S(\rho) = -\mathrm{Tr}(\rho \log{\rho}) $ denotes the von Neumann entropy. Mathematically, maximizing the von Neumann entropy under the constraint that $ \langle I \rangle = \mathrm{Tr}(\rho I) $ for quantum states $\rho$ is a standard textbook task \cite{boyd_vandenberghe_2004}. Physically, we find that the relative entropy of purity shows a distinct qualitative behaviour with respect to the incompatibility quantifier $ C$ from all the other considered resource measures. More precisely, the relative entropy of purity is not only minimized by different states, it also monotonically decreases with increasing $C$ as depicted in Fig. (\ref{Plot_rel_entr_pur}).
\begin{figure}[h!]
\includegraphics[scale =0.35]{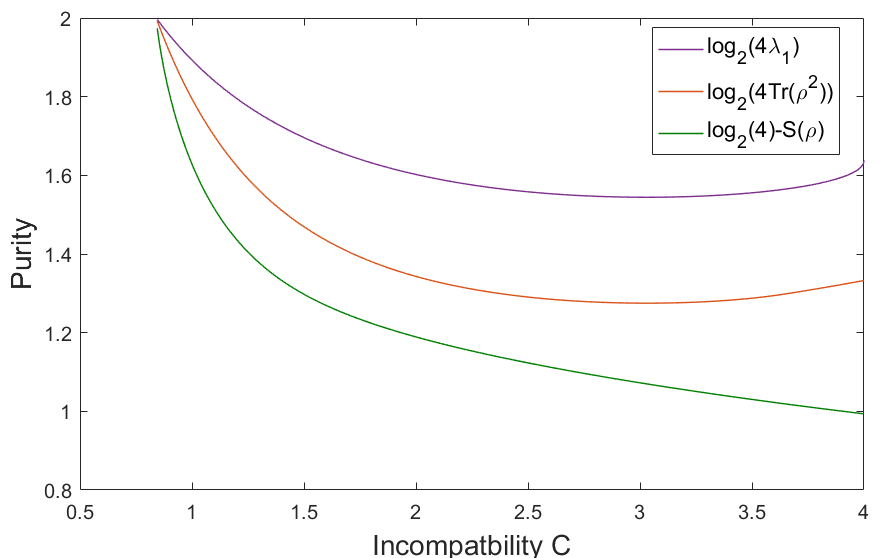}
\caption{Comparison of the logarithm of the generalized robustness of purity, the so-called log-robustness (purple) with the Rényi $2$-purity $\mathcal{P}_2 $ (orange), and the relative entropy of purity (green) for an exemplary violation of $ v = 0.2 $. The minimal relative entropy decreases with increasing $C$, for all $C$, contrary to the other measures as described in the text.}
\label{Plot_rel_entr_pur}
\end{figure} 

\end{document}